\documentclass[11pt, letter]{article}
\usepackage[english]{babel}
\usepackage[utf8x]{inputenc}
\usepackage[T1]{fontenc}

\usepackage[top=3.5cm,bottom=3cm,left=3.5cm,right=3.5cm,marginparwidth=2cm]{geometry}

\usepackage{fullpage}

\usepackage{mathrsfs}
\usepackage{amsmath}
\usepackage{amssymb}
\usepackage{graphicx}
\usepackage[colorinlistoftodos]{todonotes}

\usepackage{amsthm}
\usepackage{hyperref}
\hypersetup{colorlinks=true, allcolors=blue}
\usepackage[capitalise]{cleveref}

\newtheorem{theorem}{Theorem}[section]
\newtheorem{lemma}[theorem]{Lemma}

\newtheorem{claim}[theorem]{Claim}
\newtheorem{observation}[theorem]{Observation}

\usepackage{fancyhdr,lastpage}
\usepackage{amsfonts}

\usepackage{color}

\usepackage{pgfplots}

\usepackage{mathrsfs}
\usetikzlibrary{arrows}

\usepackage{enumerate}

\author{Jack Dippel\thanks{Department of Mathematics \& Statistics, McGill University: jack.dippel@mail.mcgill.ca}\\
\and Adrian Vetta\thanks{School of Computer Science and Dept. of Mathematics \& Statistics, McGill University: adrian.vetta@mcgill.ca}}

\title{An Improved Bound for the\\[0.2cm] Tree Conjecture in Network Creation Games}

\date{}

\begin{document}

\pagestyle{plain}

\maketitle

\vspace{-.7cm}

\begin{abstract}
We study Nash equilibria in the network creation game of Fabrikant et al.~\cite{FLM03}.
In this game a vertex can buy an edge to another vertex for a cost of $\alpha$, and the objective of
each vertex is to minimize the sum of the costs of the edges it purchases plus the sum of the
distances to every other vertex in the resultant network.
A long-standing conjecture states that if $\alpha\ge n$ then every Nash equilibrium
in the game is a spanning tree. We prove the conjecture holds for any $\alpha>3n-3$. 
\end{abstract}

\section{Introduction}

In the {\em network creation game}, there is a set $V=\{1,2,\dots, n\}$ of agents (vertices).
A vertex can buy (build) an edge for a fixed cost of $\alpha$. Thus, 
a strategy for vertex $v$ is a set of incident edges $E_v$ that it buys.
A strategy profile $\mathcal{E}=\{E_1,E_2,\dots, E_n\}$ of the agents induces a graph $G = (V,E)$, where $E=E_1\cup E_2\cup \cdots  \cup E_n$.
The objective of each vertex is to minimize its total cost which is the sum of its building cost
and its connection cost. The {\em building cost} for vertex $v$ is $\alpha\cdot |E_v|$, the cost of all the edges it buys.
The {\em connection cost} is $\sum_{u: u\neq v} d_G(u,v)$, the sum of the distances in $G$ of $v$ to every other vertex.
That is
$$c_v(\mathcal{E}) = \alpha\cdot |E_v| + \sum_{u: u\neq v} d_G(u,v)$$

Our focus is on Nash equilibria of the game.
A  strategy profile $\mathcal{E}=\{E_1,E_2,\dots, E_n\}$ is a Nash equilibrium
if $E_v$ is a best response to $\mathcal{E}_{-v}=\{E_u\}_{u\neq v}$, for every vertex $v$.
That is, no agent can reduce its total cost by buying a different set of edges, given the strategies 
of the other agents are fixed.
By default, $d_G(u,v)=\infty$ if there is no path from $u$ to $v$ in the network. It immediately follows
that every Nash equilibrium $G$ is a connected graph.
The main result in this paper is that if $\alpha > 3n-3$ then every Nash equilibrium in the
network creation game is a spanning tree.

\subsection{Background}
Motivated by a desire to study network formation in the internet, Fabrikant et al.~\cite{FLM03} introduced the {\em network creation game} in 2003. 
They proved the {\em price of anarchy} for the game is $O(\sqrt{\alpha})$, and conjectured that it is a constant.
More specifically, they showed that any Nash equilibrium that is a spanning tree has total cost at most $5$ times that of the optimal
network (a star). They then conjectured that, for $\alpha$ greater than some constant, every Nash equilibrium 
is a spanning tree. This was the original {\em tree conjecture} for network creation games.

Subsequently, there has been a profusion of work on the tree conjecture. 
Albers et al. \cite{AEE14} proved that the tree conjecture holds for any $\alpha \geq 12n \log n$. 
However, they also showed that the conjecture is false in general. Moreover, Mamageishvili et al.~\cite{MMM15}
proved the conjecture is false for $\alpha < n$.
This has lead to a {\em revised tree conjecture}, namely that every Nash equilibrium is a tree if $\alpha\ge n$.
Indeed, Mihal\'ak and Schlegel~\cite{MS13} proved the conjecture holds when $\alpha\ge c\dot n$ for a large enough constant
$c$, specifically for $c> 273$. This constant has subsequently been improved in a series of works.
Mamageishvili et al.~\cite{MMM15} proved the revised tree conjecture for $\alpha \geq  65n$ and
\`Alvarez and Messegu\'e~\cite{AM17} improved the bound to $\alpha \geq  17n$. The current best bound for the 
revised tree conjecture was given by Bil\`o and Lenzner~\cite{BL19}. They proved every Nash equilibria in the 
network creation game is a spanning tree if $\alpha > 4n -13$.

As stated, our contribution is to further improve the bound. In particular, we prove the revised tree conjecture holds 
for any $\alpha > 3n-3$. Our proof exploits the concept of a {\em min-cycle}, introduced by Lenzner in~\cite{L14}.

We remark that numerous extensions and variations of the network creation game have also been studied; 
we refer the interested reader to~\cite{ADH13, BGL16, BGL14, CLM17, CL15, DHM12}.


\section{Preliminaries}
In this section, we present structural properties of Nash equilibria in the network creation game
and introduce the main strategic tools that will subsequently be used to quantitively analyze equilibria. 
We begin with some basic notation.
Given a subgraph $H$ of a graph $G$, we let $d_H(v,u)$ denote the distance between $v$ and $u$ in $H$.
In particular, $d(u,v)=d_G(u,v)$ is the distance from $v$ to $u$ in the whole graph. We let $D(v)=\sum_{u: u\neq v} d_G(u,v)$
be the sum of the distances from $v$ to every other vertex; that is, $D(v)$ is the connection cost for vertex $v$.
The shortest path tree rooted at $v$ is denoted $T_v$.
For any vertex $u$ we denote the subtree of $T_v$ rooted at $u$ by $T_v(u)$; similarly for any edge $e\in T_u$, 
we denote the subtree of $T_v$ below $e$ by $T_v(e)$.

Recall that $G$ is an undirected graph. 
Thus, once built, an edge $(u,v)$ of $G$ can be traversed in either direction.
However, it will sometimes be useful to also view $G$ as a directed graph. Specifically,
$(u,v)$ is oriented from $u$ to $v$ if the edge was bought by $u$
and is oriented from $v$ to $u$ if the edge was bought by $v$.

\subsection{Min-Cycles}
Of primary importance is the concept of a min-cycle, introduced by Lenzner~\cite{L14} and used 
subsequently by \`Alvarez and Messegu\'e~\cite{AM17} and by Bil\`o and Lenzner~\cite{BL19}. 
A cycle $C$ is a {\em min-cycle} if $d_C(u, v) = d_G(u,v)$, for every pair of vertices $u,v \in V(C)$.  
We will require the following two known min-cycle lemmas, short proofs of which we present for completeness.
The first lemma states then min-cycles arise in any graph. 
\begin{lemma}\label{mincycleedges}\cite{L14}	
If $e$ is not a cut-edge in $G$ then the smallest cycle containing $e$ is a min-cycle.
\end{lemma}
\begin{proof}
Consider the smallest cycle $C$ containing an edge $e$. Suppose for the sake of contradiction 
that there are two vertices $u,v\in C$ such $d_G(u,v) < d_C(u,v)$. Without loss of generality, 
suppose the shortest path between $u$ and $v$, labelled $P$, lies entirely outside $C$. Note that $C$ contains two paths
between $u$ and $v$. Let $Q$ be the path of $C$ from $u$ to $v$ that contains $e$. 
Then $P \cup Q$ is a cycle containing $e$ that is strictly than $C$, a contradiction.
\end{proof}
Notice that Lemma~\ref{mincycleedges} applies to any graph $G$ regardless of whether or not it is a Nash equilibrium.
The second lemma states that, given the orientations of the edges, every min-cycle in a Nash equilibria is directed.
To prove this, we require one definition and an observation. Given a min-cycle $C$, we say that edge $\bar{e}\in C$ is {\em opposite} 
$v\in C$ if it is at least as far from $v$ as
any other edge in $C$. Note that if $C$ has odd cardinality then there is a unique edge opposite to  $v\in C$; if
$C$ has even cardinality then there are two edges opposite to $v\in C$. Our observation is then:

\begin{observation}\label{obs:trees}
For any vertex $v$ in a min-cycle $C$, every edge of $C$ is in the shortest path tree $T_{v}$ except for 
one edge opposite $v$.
\end{observation}
\begin{proof}
If $C$ has even cardinality then there is a unique vertex $u \in C$ farthest from $v$.
Tthe two paths in $C$ from $v$ to $u$ have the same length. 
Only one of these two paths is needed in $T_{v}$. In particular, either of the edges in $C$ incident $u$
(thus, opposite $v_i$) may be excluded from $T_{v}$. Since, $C$ is a min-cycle the remaining edges of $C$ are in 
some $T_{v}$.

If $C$ has odd cardinality then there are two vertices $u_1, u_2\in C$ furthest from $v$.
Then the edge $(u_1,u_2)$ is opposite $v$ and is not in $T_{v}$.
Again, since, $C$ is a min-cycle the remaining edges of $C$ are in 
some $T_{v}$.
\end{proof}

\begin{lemma}\label{directedcycles}\cite{BL19}
Let $\alpha > 2(n-1)$. Then every min-cycle in a Nash equilibrium $G$ is directed.
\end{lemma}
\begin{proof}
Let $C$ be a min-cycle that is {\em not} directed. Then there is a vertex $v$ that buys two 
edges of the cycle, say $e_1$ and $e_2$. Take an edge $\bar{e}$ opposite $v$ in $C$,
and let $u\in C$ be the vertex which buys $\bar{e}$.
By Observation~\ref{obs:trees}, a shortest path tree $T_{v}$ rooted at $v$ need not contain 
$\bar{e}$. Thus $u$ can sell $\bar{e}$ and buy the edge $uv$ without increasing 
the distance from $v$ to any other vertex. It follows, by the Nash equilibrium conditions, that 
\begin{equation}\label{eq:directed-cycle-1}
D(u) \leq D(v) + n-1
\end{equation}
On the other hand, the edges $e_1$ or $e_2$ are not needed in a shortest path from $u$ to any
other vertex $w\in C\setminus\{v\}$. It follows that $v$ can sell {\em both} $e_1$ and $e_2$ and instead buy the edge $vu$
without increasing the distance from $u$ to any other vertex. 
It follows, by the Nash equilibrium conditions, that 
\begin{equation}\label{eq:directed-cycle-2}
D(v) \leq D(u) + n-1 - \alpha
\end{equation}
Together (\ref{eq:directed-cycle-1}) and  (\ref{eq:directed-cycle-2}) give
$$D(v) \ \leq\  D(v) + 2(n-1) - \alpha \ <\   D(v)$$
This contradiction implies that $C$ must be a directed cycle.
\end{proof}

Given \cref{directedcycles}, we may use the following specialized notation for min-cycles.
We will label the vertices of a min-cycle as $C=\{v_0,v_1,\dots, v_k=v_0\}$, where $k=|C|$.
As the min-cycle is directed, we may also assume $v_i$ buys the edge $e_i=v_iv_{i+1}$, for 
each $0 \leq i \leq |C|-1$. Recall there are two paths on $C$ between any pair of vertices $v_i$ and $v_j$. 
The path that follows the orientation of the edges in the directed cycle is called the {\em clockwise path} from $v_i$ to $v_j$ and 
that path that goes against the orientation of the cycle is called the {\em anti-clockwise path}. 

To conclude our discussion on min-cycles, we present two trivial, but very useful, observations.
First we remark that every min-cycle is {\em chordless}. 
\begin{observation}\label{obs:chordless}
Every min-cycle $C$ is chordless.
\end{observation}
\begin{proof}
This follows immediately from the definition of a min-cycle.
\end{proof}
\cref{obs:chordless} implies that any vertex $v_i\in C$ has exactly two neighbours 
on $C$, namely $v_{i-1}$ and $v_{i+1}$. Any other neighbour of $v_i$ must lie outside of $C$.
For the second observation, define $k_{\max}$ to be the maximum length of a min-cycle in a biconnected 
component $H$.\footnote{A {\em biconnected 
component} $H\subseteq G$ is
a maximal set such that there are two vertex-disjoint paths between any pair of vertices in $H$.}

\begin{observation}\label{obs:increase}
Let $e=(u,v)$ be an edge in a biconnected component $H$ of $G$. Then 
$$d_{G\setminus e}(u,v) \ \le\ d_{G}(u,v)+ k_{\max}-2$$
\end{observation}
\begin{proof}
Since $H$ is biconnected, $e$ lies in a cycle. Thus, by \cref{mincycleedges}, $e$ lies in a min-cycle $C$.
Because there are two paths between $u$ and $v$ on $C$ (the clockwise and anti-clockwise paths),
the removal of $e=(u,v)$ increases the distance between $u$ and $v$ by at most $(|C|-1)-1\le k_{\max}-2$.
The observation follows.
\end{proof}
The applicability of Observation~\ref{obs:increase} is evident. 
It can be used to upper bound the increase in connection costs arising from the sale of an edge.
At a Nash equilibrium, this upper bound must be at least $\alpha$, the amount saved in
construction costs by the sale.

\subsection{Basic Strategic Options for a Vertex}
To prove a Nash equilibrium $G$ is a tree we must show $G$ contains no cycles.
In particular, it suffices to show $G$ contains no biconnected components.
This leads to the following basic idea used in our proof.
We take a vertex $v$ in a biconnected component $H$ of a Nash equilibrium $C$.
By the Nash equilibrium conditions, any strategy change by $v$ cannot decrease its total cost.
Thus the resultant change in cost is non-negative. In particular, each potential strategy change induces an
{\em inequality constraint}. Further, a collection of potential strategy changes induces a set of inequalities.
By case analysis, we will show that it is always possible to find a collection of strategy changes 
for which linear combinations of the corresponding inequality constraints 
induce the constraint $\alpha\le 3n-3$. It follows that if $\alpha> 3n-3$ then $G$ is not a Nash equilibrium.
Consequently, every Nash equilibrium for $\alpha> 3n-3$ has no biconnected components and is, thus, 
a spanning tree.

In fact, nearly every strategy change we use in the paper will take one of three general forms.
Therefore, it will be helpful to present these three general forms here.

\begin{claim}[Strategy I]\label{cl:strategy-1}
Let $H$ be a biconnected component in a Nash equilibrium $G$ with $u,v\in H$.
Assume $v$ buys edge $vw \in H$. 
If $w \not\in T_u(v)$, for some $T_u(v)$, then the strategy change where $v$ swaps $vw$ for $vu$
induces the inequality
\begin{equation}\label{eq:strat1}
	0 \ \leq\  D(u) + (n-1) - D(v) - (d(u,v)+1)\cdot|T_u(v)|
\end{equation}
\end{claim}
\begin{proof}
This strategy change is illustrated in Figure~\ref{fig:I}.
The basic idea is that by adding the edge $vu$ the vertex $v$ can now use the shortest path tree $T_u$. 
The distance from $v$ to any vertex $x$ would then be at most $1+d(u,x)$, inducing 
a communication cost of at most $D(u)+(n-1)$ instead of its current communication cost of $D(u)$.

By selling $vw$ the construction cost for vertex $v$ does not change. But there is a problem. What if
$vw$ is used in the shortest path tree $T_u$? 
Since $w\notin T_u(v)$, the only way this can occur is if $wv$ is the final edge on the path $P$ from $u$ to $v$ in $T_u$.
That is, the edge $wv$ is only needed by $u$ in order to connect to the vertices in $T_u(v)$. But it can now connect 
to those vertices simply by using the edge $uv$ instead of the path $P$.
In particular, by swapping $vw$ for $vu$: (i) the distance form $u$ to any other vertex $x$ cannot increase,
and (ii) the distance from $u$ to any vertex $x\in T_u(v)$ decreases by $|P|-1=d(u,v)-1$.
Moreover, $v$ does not need to use the edge $vu$ in order to connect to the vertices in $T_u(v)$.
Thus the total change in cost is at most
\begin{alignat*}{3}
\sum_{x\notin T_u(v)} \big(1+d(u,v)\big) &+ \sum_{x\in T_u(v)} \big(d(u,x)-|P|\big) -D(v)\\
&= D(u) + \big(n-|T_u(v)|\big)  - |P|\cdot |T_u(v)| -   D(v)\\
&= D(u) + \big(n-1\big)  - (|P|+1)\cdot |T_u(v)| -   D(v)\\
&= D(u) + \big(n-1\big)  - \big(d(u,v)+1\big)\cdot |T_u(v)| -   D(v)
\end{alignat*}
But, by the Nash equilibrium conditions, this change in cost must be non-negative.
The claimed inequality (\ref{eq:strat1}) then holds.
\end{proof}

\begin{figure}[t]
 \begin{minipage}[t]{0.32\textwidth}		
		\centering
	\scalebox{1}{
		\begin{tikzpicture}
		 \draw [green, ultra thick] (0,2) -- (0,0);
			\filldraw (0,2) circle (0.05);
			\filldraw (0,0) circle (0.05);
			\filldraw (.5,0.25) circle (0.05);
			\draw[thick] (-1.75,-1.85) -- (1.75,-1.85)-- (1.75, 0) -- (0,2)-- (-1.75, 0) -- (-1.75,-1.85);
			\draw[dashed, thick, red] (-1.95,-1.15) -- (1.95,-1.15)-- (1.95, 0) -- (0,2.2)-- (-1.95, 0) -- (-1.95,-1.15);
			\draw[thick, blue] (.75,-1.75) -- (-.75,-1.75) -- (0,0)--(.75,-1.75) ;
			\draw[thick, dashed] (0,0) -- (.5,.25);
			\draw(.5,0.25) -- (.25,0.6) -- (.65,.95)-- (.15, 1.3)-- (0.2,1.65) --(0,2);
			\begin{scriptsize}
				\draw[color=red] (-2.15,0) node {\footnotesize $H$};
				\draw (-.8,.5) node {\footnotesize $T_u$};
				\draw(-.15, 0) node {$v$};
				\draw (0.2,2.15) node {$u$};
				\draw (0.7,0.25) node {$w$};
				\draw (0.3,.9) node {$P$};
				\draw[color=blue] (0,-1.5) node {\footnotesize $T_u(v)$};
			\end{scriptsize}
		\end{tikzpicture}}
	\caption{Strategy I}\label{fig:I}
	\end{minipage}
 \begin{minipage}[t]{0.32\textwidth}
 	\centering	
\scalebox{1}{
	\begin{tikzpicture}
		 \draw [green, ultra thick] (0,2) -- (0,0.5);
			\filldraw (0,2) circle (0.05);
			\filldraw (0,.5) circle (0.05);
			\filldraw (.5,0.75) circle (0.05);
			\draw[thick] (-1.75,-1.85) -- (1.75,-1.85)-- (1.75, 0) -- (0,2)-- (-1.75, 0) -- (-1.75,-1.85);
			\draw[dashed, thick, red] (-1.95,-1.15) -- (1.95,-1.15)-- (1.95, 0) -- (0,2.2)-- (-1.95, 0) -- (-1.95,-1.15);
			\draw[thick, blue] (.75,-1.75) -- (-.75,-1.75) -- (0,-.5)--(.75,-1.75) ;
			\draw[thick, dashed] (0,0.5) -- (.5,.75);
			\draw(.5,.75) -- (.65,.95)-- (.15, 1.3)-- (0.2,1.65) --(0,2);
			\begin{scriptsize}
				\draw[color=red] (-2.15,0) node {\footnotesize $H$};
				\draw (-.8,.5) node {\footnotesize $T_u$};
				\draw(-.15, 0.65) node {$v$};
				\draw (0.2,2.15) node {$u$};
				\draw (0.7,0.75) node {$w$};
				\draw (0.3,.95) node {$P$};
				\draw[color=blue] (0,-1.5) node {$T_u(e)$};
				\draw (-.65,0) node {$C$};
				\draw (.15,0) node {$e$};
			\end{scriptsize}
	\filldraw (0,-.5) circle (0.05);
		\draw[thick] (0,0.5) arc (90:270:.5);
		\draw[dashed, thick] (0,0.5) -- (0,-.5);
\end{tikzpicture}}
	\caption{Strategy II}\label{fig:II}
\end{minipage}
 \begin{minipage}[t]{0.32\textwidth}
		\centering
	\scalebox{1}{
		\begin{tikzpicture}
		 \draw [green, ultra thick] (0,2) -- (0,0);
			\filldraw (0,2) circle (0.05);
			\filldraw (0,0) circle (0.05);
			\filldraw (.5,0.25) circle (0.05);
			\draw[thick] (-1.75,-1.85) -- (1.75,-1.85)-- (1.75, 0) -- (0,2)-- (-1.75, 0) -- (-1.75,-1.85);
			\draw[dashed, thick, red] (-1.95,-1.15) -- (1.95,-1.15)-- (1.95, 0) -- (0,2.2)-- (-1.95, 0) -- (-1.95,-1.15);
			\draw[thick, blue] (0,0) -- (1,-1.75) -- (-1,-1.75) -- (-.25,-.5) ;
			\draw[thick, dashed] (0,0) -- (.5,.25);
			\draw(.5,0.25) -- (.25,0.6) -- (.65,.95)-- (.15, 1.3)-- (0.2,1.65) --(0,2);
			\begin{scriptsize}
				\draw[color=red] (-2.15,0) node {\footnotesize $H$};
				\draw (-.8,.5) node {\footnotesize $T_u$};
				\draw(-.15, 0) node {$v$};
				\draw (0.2,2.15) node {$u$};
				\draw (0.7,0.25) node {$z$};
				\draw (-0.05, -.5) node {$x$};
				\draw (-0.8, -1) node {$w$};
				\draw (-.6, -1.5) node {$c$};
				\draw (-0.05, -1) node {$y$};
				\draw[color=blue] (-1.3,-1.5) node {$T_u(v)$};
					\draw (0.15,-1.35) node {\tiny $T_u(x)$};
			\end{scriptsize}
			\filldraw (0.1,-.5) circle (0.05);
			\filldraw (-.25,-.5) circle (0.05);
			\filldraw (-.4,-1.5) circle (0.05);
			\filldraw (.1,-.95) circle (0.05);
			\filldraw (-.65,-1.15) circle (0.05);
			\draw (0.1,-.5) -- (.6,-1.5) -- (-.4,-1.5)-- (0.1,-.5);
			\draw (.1,-.95) -- (0.5,.25);
			\draw [dashed, thick] (-.25,-.5) -- (0,0);
			\draw [dashed, thick] (.1,-.5) -- (0,0);
	\end{tikzpicture}}
	\caption{Strategy III}\label{fig:III}
\end{minipage}
\end{figure}

\begin{claim}[Strategy II]\label{cl:strategy-2}
Let $H$ be a biconnected component in a Nash equilibrium $G$ with $u, v\in H$.
Assume $v$ buys edges $vw, e\in H$. If $w \not\in T_u(v)$, for some $T_u(v)$, then the strategy change where $v$ swaps both
$vw$ and $e$ for $vu$ induces the inequality
\begin{equation}\label{eq:strat2}
	0 \ \leq\  D(u) + (n-1) - D(v) - (d(u,v)+1)\cdot|T_u(v)| +(k_{\max}-2)\cdot |T_u(e)| - \alpha
\end{equation}
\end{claim}
\begin{proof}
This strategy change is illustrated in Figure~\ref{fig:II}.
The first four terms of~(\ref{eq:strat2}) follow as in the argument of \cref{cl:strategy-1}.
What remains to be considered is the additional effect of selling edge $e$.
The $-\alpha$ term arises in the construction cost because $v$ is now buying one less edge than before.
Since $e$ is in a biconnected component $H$, it is in a min-cycle $C$ by \cref{mincycleedges}.
By~\cref{directedcycles},  $C$ is directed. Hence, $vw\notin C$.
The term $(k_{\max}-2)\cdot |T_u(e)|$ then follows by applying \cref{obs:increase}.
\end{proof}

Note that any vertex in $T_u(e)$ is also in $T_u(v)$. Therefore, instead of (\ref{eq:strat2}), we will often use the 
following simplified bound:
\begin{equation}\label{eq:strat3}
	0 \ \leq\  D(u) + n-1 - D(v) + (k_{\max}-d(u,v)-3)\cdot|T_u(v)| - \alpha
\end{equation}

\begin{claim}[Strategy III]\label{cl:strategy-3}
Let $H$ be a biconnected component in a Nash equilibrium $G$ with $u, v\in H$.
Assume $v$ buys a quantity $\ell$ of edges of $H$. Let $w$ be the vertex in any $T_u(v)\cap H$ furthest from $u$.
Then the strategy change where $v$ swaps all its edges for $vu$ induces the inequality
\begin{equation}\label{eq:strat4}
	0 \ \leq\  D(u) + (n-1) - D(v) + 2d(v,w)\cdot |T_u(v)| - (\ell-1)\cdot \alpha
\end{equation}
\end{claim}
\begin{proof}
This strategy change is illustrated in Figure~\ref{fig:III}.
Again, the first three terms of~(\ref{eq:strat4}) follow as in \cref{cl:strategy-1}.
The final term arises because $v$ buys $(\ell-1)$ fewer edges after the strategy change.
It remains to explain the fourth term, namely $2d(v,w)\cdot|T_u(v)|$.
Let $L$ be the set of $\ell$ edges that $v$ sells. 
This change may cause the distances from $u$ to vertices in $T_u(v)$ to increase. 

To quantify this effect, we first prove that for any pair of vertices $a,b$ with edges $va,vb\in L$, the subtrees
$T_u(a)$ and $T_u(b)$ have no edge between them. 
Suppose, for a contradiction, that vertex $t\in T_u(a)$ buys an edge to vertex $s \in T_u(b)$.
By \cref{cl:strategy-1}, we have
\begin{alignat}{1}\label{eq:cl3a}
0 &\leq D(v) + n-1 - D(t) - (d(t,v)+1)\cdot|T_v(t)|  \nonumber\\
&\leq D(v) + n-1 - D(t)
\end{alignat}
On the other hand, by \cref{cl:strategy-2}, we have
\begin{alignat}{1}\label{eq:cl3b}
0 &\leq D(t) + n-1 - D(v) - (d(t,v)+1)\cdot|T_t(v)| +(k_{\max}-2)\cdot |T_t(vb)| - \alpha  \nonumber\\
&\le  D(t) + n-1 - D(v) - \alpha
\end{alignat}
The second inequality holds because $|T_t(vb)|=0$. To see this,  
note that the shortest path from $t$ to $v$ need not use $vb$ as $t\in T_u(a)$.
Furthermore, the path from $t$ to $b$ consisting of the edge $ts$ plus the subpath from $s$ to $b$ in $T_u(b)$ is at most as long as 
the path from $t$ to $b$ using the subpath from $t$ to $v$ in $T_u(v)$ plus the edge $vb$; otherwise it would be the case that $s\in T_u(a)$. 
Thus the shortest path from $t$ to any vertex need not use $vb$ and so $|T_t(vb)|=0$.
Summing (\ref{eq:cl3a}) and (\ref{eq:cl3b}) gives $\alpha <  2(n-1)$, a contradiction. 

Next, take $vx \in L$. Since $H$ is biconnected, there is an edge $yz$ with $y\in T_u(x)$ and $z\notin T_u(x)$ where $z\neq v$ and possibly $y=x$.
Furthermore, as proven above, $z\notin T_u(b)$ for any $vb\neq vx \in L$.
Therefore the distance from $z$ to $u$ does not increase when $v$ sells its edges. In addition, the distance 
between any vertex in $T_u(x)$ and $x$ is the same as before $vx$ was sold. 
We also have $d(x,y) \leq d(v,w) -1$, by our choice of $w$, because $y \in T_u(v)$ and $y\neq v$.

We claim $d(u,z) \leq d(u,v) + d(v,w)$. If not, assume $d(u,z) \geq d(u,v) + d(v,w) +1$.
Then we could choose $T_u(v)$ such that $z \in T_u(v)$, contradicting our choice of $w$.
Therefore, for any vertex $c \in T_u(x)$, originally we have  $d(c,u) = d(u,v) + d(v,c)$. 
But, after the change in strategy of vertex~$v$, there is a path 
from $c$ to $u$, via $x, y, z$, of length at most
\begin{alignat}{1}\label{eq:cl3c}
d(c,x) + d(x,y) + d(y,z) + d(z,u) &\leq (d(v,c)-1) + (d(v,w)-1) +1 +  \left( d(u,v) + d(v,w)\right) \nonumber \\
& = (d(c,v)+d(v,u)) -1 + 2d(v,w) \nonumber\\
& < d(c,u) + 2d(v,w)
\end{alignat}
Thus the length of the shortest path from $u$ to $c$ increases by less that $2d(v,w)$. This gives the fourth 
term~(\ref{eq:strat4}).
\end{proof}

This completes the description of the three main strategic options that we will study.
We remark that we will apply these strategies to vertices on a min-cycle.
This is valid because every biconnected component contains a min-cycle.
\begin{observation}\label{obs:biconn}
Any biconnected component $H$ of cardinality at least two contains a min-cycle~$C$.
\end{observation}
\begin{proof}
As $H$ has cardinality at least two it contains an edge $e=(u,v)$.
Note that $e$ cannot be a cut edge or $H$ is not biconnected. Thus by \cref{mincycleedges}, $e$ is in a 
min-cycle $C$. It immediately follows that $C\subseteq H$.
\end{proof}

\section{Equilibria Conditions}

In this section, we derive further structural properties that must be satisfied at a Nash equilibrium.

\subsection{Biconnected Components}
We have already derived some properties of min-cycles in a Nash equilibrium $G$.
Recall, that we wish to prove the non-existence of biconnected components in a Nash equilibrium.
So, more generally, we will now derive properties satisfied by any biconnected component $H$ in a Nash equilibrium.

The girth of a graph $G$, denoted $\gamma$, is the length of its smallest cycle.\footnote{The girth is infinite if $G$ is a forest.}
The following lemma that lower bounds the girth of an equilibrium will subsequently be useful.
\begin{lemma}\label{smallestcycle}
	Let $\alpha > 2(n-1)$. Then the girth of any Nash equilibrium $G$ satisfies $\gamma(G) \geq \frac{2\alpha}{n-1} + 2$.
\end{lemma}
\begin{proof}
	Take a minimum length cycle $C = v_0, v_1\dots, v_{k-1}, v_{k}=v_0$ in $G$. 
	By \cref{mincycleedges}, $C$ is a min-cycle. Therefore, By \cref{directedcycles}, $C$ is a directed cycle.
	So, as stated, we may assume $e_i=v_iv_{i+1}$ is bought by $v_i$, for each $0 \leq i \leq k-1$. 
	Now, for each vertex $u\in V$, we define a set $L_u\subseteq \{0,1,\dots, k-1\}$ as follows. We have $i\in L_u$ if and only if
	{\em every} shortest path from $v_i\in C$ to $u$ uses the edge $e_i$.
	(In particular, $u\in T_{v_i}(e_i)$ for every shortest path tree $T$ rooted at $v_i$.)
	
	We claim $|L_u| \le \frac{|C|-1}{2}$ for every vertex $u$.
	If not, take a vertex $u$ with $|L_u| > \frac{|C|-1}{2}$. 
	Let $d(v_i,u)$ be the shortest distance between $u$ and $v_i\in C$. 
	Next give $v_i$ a label $\ell_i= d(v_i,u)-d(v_{i+1},u)$. Observe that $\ell_i\in \{-1,0,1\}$.
	Furthermore, the labels sum to zero as
	$$\sum_{i=0}^{|C|-1} \ell_i \ =\  \sum_{i=0}^{|C|-1} \left( d(v_i,u)-d(v_{i+1},u) \right) 
		\ =\ \sum_{i=0}^{|C|-1} d(v_i,u) - \sum_{i=1}^{|C|} d(v_{i},u)  \ = \ 0$$
	Now take a vertex $v_i$ in $C$ that uses $e_i$ in {\em every} shortest path to $u$; that is, $i\in L_u$.
	Then $\ell_i=1$ and $\ell_{i-1}\ge 0$. 
	In particular, if $|L_u| > \frac{|C|-1}{2}$ then there are $> \frac{|C|-1}{2}$ positive 
	labels and $> 1+\frac{|C|-1}{2}$ non-negative labels.
	Hence, there are $< |C|-\left(1+\frac{|C|-1}{2}\right) = \frac{|C|-1}{2}$ negative labels.
	But then the sum of the labels is strictly positive, a contradiction.
	
	Now, for each $i$, let $T_{v_i}$ be a shortest path tree rooted at $v_i$ such that the size of $T_{v_i}(e_i)$ is minimized.
	As $|L_u| \le \frac{|C|-1}{2}$ for every vertex $u$, there exists a $v_i$ with  $|T_{v_i}(e_i)\setminus C| \leq \frac{n-|C|}{2}$.
	On the other hand, clearly $|T_{v_j}(e_j)\cap C| \leq \frac{|C|-1}{2}$ for every $v_j$.
	It follows that $|T_{v_i}(e_i)|\leq \frac{n-1}{2}$. But then if $v_i$ sells $e_i$ its cost increases by at most
	$(|C|-2)\cdot |T_{v_i}(e_i)| - \alpha \leq (|C|-2)\cdot \frac{n-1}{2} - \alpha$. 
	This must be non-negative by the Nash equilibrium conditions.
	Rearranging, this implies that $\frac{2\alpha}{n-1} + 2\leq |C|$ as desired.
\end{proof}
We remark that \cref{smallestcycle} actually holds for all $\alpha$. We omit the proof of this fact as we 
only need the result for the case of $\alpha > 2(n-1)$.

\begin{lemma}\label{ceilingfloorequal}
	In a min cycle $C$, $T_{v_{\lfloor\frac{k}{2}\rfloor}}(v_0)= T_{v_{\lceil\frac{k}{2}\rceil}}(v_0)$ for some choice 
	of $T_{v_{\lfloor\frac{k}{2}\rfloor}}(v_0)$ and $T_{v_{\lceil\frac{k}{2}\rceil}}(v_0)$.
\end{lemma}

\begin{proof}
	For even $k$, $\lfloor\frac{k}{2}\rfloor = \lceil\frac{k}{2}\rceil$, and the result is trivial.
	For odd $k$, suppose $T_{v_{\lfloor\frac{k}{2}\rfloor}}(v_0) \not= T_{v_{\lceil\frac{k}{2}\rceil}}(v_0)$. 
	Without loss of generality, let $x \in T_{v_{\lfloor\frac{k}{2}\rfloor}}(v_0) \setminus T_{v_{\lceil\frac{k}{2}\rceil}}(v_0)$. 
	Therefore there is a shortest 
	path $P$ from $x$ to $v_{\lfloor\frac{k}{2}\rfloor}$ which goes through $v_0$. Hence, there must be a path of the same 
	length from $x$ to  $v_{\lceil\frac{k}{2}\rceil}$. Thus either $x$ may belong to $T_{v_{\lceil\frac{k}{2}\rceil}}(v_0)$ for a 
	different choice of $T_{v_{\lceil\frac{k}{2}\rceil}}(v_0)$, or there is a path $Q$ between $x$ and $v_{\lceil\frac{k}{2}\rceil}$ that 
	is strictly shorter than $P$ and does not use $v_0$. But then $Q \cup v_{\lceil\frac{k}{2}\rceil}v_{\lfloor\frac{k}{2}\rfloor}$ 
	is a path from $x$ to $v_{\lfloor\frac{k}{2}\rfloor}$ of length at most $|P|$ which does not use $v_0$.
	This implies we can choose $T_{v_{\lfloor\frac{k}{2}\rfloor}}(v_0)$ 
	such that $x \not\in T_{v_{\lfloor\frac{k}{2}\rfloor}}(v_0)$. It follows that we
	may choose $T_{v_{\lfloor\frac{k}{2}\rfloor}}(v_0) = T_{v_{\lceil\frac{k}{2}\rceil}}(v_0)$, as desired. 
\end{proof}

Given a biconnected component $H$ and a vertex $v\in H$, let $S_H(v)$ 
be the set of vertices in $G$ that are closer to $v$ than to any other vertex of $H$.
\begin{lemma}\label{lem:SH}
For any pair of vertices $u$ and $v$ in a biconnected component $H$, we have $S_H(u)\cap S_H(v)=\emptyset$.
\end{lemma}
\begin{proof}
For a contradiction, take $x\in S_H(u)\cap S_H(v)$. 
Note that $x\in S_H(x)$ and no other $S_H(v)$ for any $x\in H$. Thus, it must be that $x\notin H$.
We may assume that $x$ is the closest vertex to $u$ and $v$ in $S_H(u)\cap S_H(v)$.
In particular, there are shortest paths $P$ from $x$ to $u$ and $Q$ from $x$ to $v$ 
that are disjoint except for their source $x$. But then $H\cup P\cup Q$ is biconnected,
contradicting the maximality of $H$.
\end{proof}

\begin{lemma}\label{biconnected}
Let $\alpha > 2(n-1)$. Then any min-cycle $C$ of a biconnected component $H$ has a vertex $v\in C$ 
which buys an edge $e \in H\setminus C$.
\end{lemma}
\begin{proof}
	First, assume that $C=H$. By \cref{smallestcycle}, $C=\{v_0,v_1,\dots, v_k=v_0\}$ has length 
	$\frac{2\alpha}{n-1}+2 > 6$ as $\alpha > 2(n-1)$.
	By \cref{directedcycles}, $v_i$ buys the edge $e_i=v_iv_{i+1}$ for $0 \leq i \leq k-1$.
	Without loss of generality, let $v_1  = argmin_{v \in C}\, |S_H(v)|$. 
	Then suppose $v_0$ sells $e_0$ and buys $v_0v_2$. This reduces its costs by at least $|S_H(v_2)| + |S_H(v_3)| -  |S_H(v_1)| > 0$, 
	contradicting the Nash equilibrium conditions. 
	
	Second, assume $C \not=  H$. Take an edge $f \in H\setminus C$ incident to a vertex of $C$. 
	As $H$ is biconnected, $f$ is in a cycle $D$ which, by \cref{mincycleedges}, we may assume is a min-cycle.
	Furthermore, by \cref{directedcycles}, $D$ is directed. But $D$ contains $f$ and so intersects $C$. Therefore, 
	there must be a vertex of $C$  which buys an edge $e \in D\setminus C$ as desired.
\end{proof}

\subsection{The Key Lemma}

The following lemma will be critical in proving the main result.
\begin{lemma}\label{multipleoutedge}
Let $\alpha > 2(n-1)$. In a maximum length min-cycle $C$ of a biconnected component $H$, there exist two 
vertices $u,v\in C$ with $d(u,v)\geq \frac{k_{\max}}{3}$ which buy edges $f, g\not\in C$ respectively. 
\end{lemma}

\begin{proof}
Take a maximum length min-cycle $C$. By \cref{biconnected}, we may assume $v_0$ buys an edge $f$ outside $C$.  
Suppose, for the sake of contradiction, only vertices in $W = \{ v_0,\dots v_{\lceil\frac{k_{\max}}{3}\rceil-1}\}$ buy an edge outside $C$. 
Let $X =  \{ v_{\lceil\frac{k_{\max}}{3}\rceil},\dots v_{\lfloor\frac{k_{\max}}{2}\rfloor}\}$. Since $\alpha > 2(n-1)$, \cref{smallestcycle} 
implies that the girth $\gamma(G)\ge 7$, which means both $X$ and $W$ are non-empty.
We now break the proof up into three cases: 

\

\noindent{\bf Case 1: There exists $x \in X$ with $deg_H(x) \geq 3$}\\
As $x \not\in W$ there is an edge $d\in H$ incident to $x$ that is bought by a vertex outside $C$. 
Since $d$ is in $H$ it is not a cut-edge. Let $D$ be a minimum length cycle $D$ containing $d$. Thus $D$ is a directed 
min-cycle by \cref{mincycleedges} and~\cref{directedcycles}. 

Since $D$ is directed it contains a vertex $w\in W$ which buys an edge of $D\setminus C$. 
Take $w$ to be the last vertex of $C$ before $x$ in $D$.
Let $\overrightarrow{D}_{wx}$ be the clockwise path in $D$ from $w$ to $x$ and let
$\overrightarrow{D}_{wx}$ be the anticlockwise path.
Next observe that the clockwise path $P$ from $w$ to $x$ in $C$ is shorter than the anticlockwise part.
Thus, because $C$ is a min-cycle, $P$ is a shortest path from $w$ to $x$.
But then $|P\cup \overrightarrow{D}_{wx}| \le |D|$. Furthermore, $P\cup \overrightarrow{D}_{wx}$ is a cycle containing $d$.
So $P\cup \overrightarrow{D}_{wx}$ is also a minimum length length cycle containing $d$.
But then it must be a directed cycle; this contradicts the fact that $P$ and $\overrightarrow{D}_{wx}$ are both paths
directed from $w$ to $x$.

\

\noindent{\bf Case 2: $deg_H(v_{\lfloor\frac{k_{\max}}{2}\rfloor+1}) \geq 3$}\\
Note that $v_{\lfloor\frac{k_{\max}}{2}\rfloor+1} \not\in W$. So there is some vertex $u\in H\setminus C$ which buys 
an edge $d = uv_{\lfloor\frac{k_{\max}}{2}\rfloor+1}$. By \cref{mincycleedges} and \cref{directedcycles}, $d$ is in a directed 
min cycle $D \not= C$. Therefore, some vertex $w\in W$ must buy an edge $f \in D\setminus C$.
If $w\neq v_0$ then exactly the same argument as in Case 1 can be applied. Thus, the only possibility remaining is $w=v_0$. 
However, by~\cref{obs:trees} applied to $C$, the shortest path tree $T_{v_{\lfloor\frac{k_{\max}}{2}\rfloor+1}}$ rooted 
at $v_{\lfloor\frac{k_{\max}}{2}\rfloor+1}$ need not contain the edge $e_0$.
Similarly, as $D$ has length at most that of $C$, the shortest path tree $T_{v_{\lfloor\frac{k_{\max}}{2}\rfloor+1}}$ 
need not contain the edge $f$ either.

Now apply Strategy II. Noting that $T_{v_{\lfloor\frac{k_{\max}}{2}\rfloor+1}}(f) = \emptyset$, 
we have that (\ref{eq:strat2}) gives
\begin{equation}\label{eq:bi-2-a}
D(v_{\lfloor\frac{k_{\max}}{2}\rfloor+1})- D(v_0) +n-1  -\alpha \ge 0
\end{equation}

On the other hand, suppose we apply Strategy I with $v_{\lfloor\frac{k_{\max}}{2}\rfloor+1}$ selling $e_{\lfloor\frac{k_{\max}}{2}\rfloor+1}$ 
and buying $v_{\lfloor\frac{k_{\max}}{2}\rfloor+1}v_0$. By (\ref{eq:strat1}), we have
\begin{equation}\label{eq:bi-2-b} 
D(v_0) - D(v_{\lfloor\frac{k_{\max}}{2}\rfloor+1}) +n-1   \ge 0
\end{equation}
Together, (\ref{eq:bi-2-a}) and (\ref{eq:bi-2-b}) imply $\alpha \leq 2(n-1)$, a contradiction.

\

\noindent{\bf Case 3: Else}\\
So $deg_H(v) =2$ for every vertex $v \in X\cup v_{\lfloor\frac{k_{\max}}{2}\rfloor+1}$.
In particular, we have $deg_H(v_{\lfloor\frac{k_{\max}}{2}\rfloor}) = deg_H(v_{\lfloor\frac{k_{\max}}{2}\rfloor+1}) = 2$.
We now consider four strategy changes.

\

\noindent (i) $v_{\lfloor\frac{k_{\max}}{2}\rfloor-1}$ sells $e_{\lfloor\frac{k_{\max}}{2}\rfloor-1}$ and 
buys $v_{\lfloor\frac{k_{\max}}{2}\rfloor-1}v_{\lfloor\frac{k_{\max}}{2}\rfloor+1}$. 

As $deg_H(v_{\lfloor\frac{k_{\max}}{2}\rfloor})=2$, the only vertices now further from $v_{\lfloor\frac{k_{\max}}{2}\rfloor-1}$
are those in $S_H(v_{\lfloor\frac{k_{\max}}{2}\rfloor})$. Their distances have increased by exactly $1$.
On the other hand, this strategy change decreases the distance from $v_{\lfloor\frac{k_{\max}}{2}\rfloor-1}$ 
to $S_H(v_{\lfloor\frac{k_{\max}}{2}\rfloor}+1)$ by $1$.
Therefore, the Nash equilibrium conditions imply
\begin{equation}\label{eq:bi-3-a} 
 |S_H(v_{\lfloor\frac{k_{\max}}{2}\rfloor}+1)| \le |S_H(v_{\lfloor\frac{k_{\max}}{2}\rfloor})|
\end{equation}

\

\noindent (ii) $v_{\lfloor\frac{k_{\max}}{2}\rfloor}$ sells $e_{\lfloor\frac{k_{\max}}{2}\rfloor}$ and 
buys $v_{\lfloor\frac{k_{\max}}{2}\rfloor}v_{\lfloor\frac{k_{\max}}{2}\rfloor+3}$. 

As $deg_H(v_{\lfloor\frac{k_{\max}}{2}\rfloor}+1)=2$, the only vertices now further from $v_{\lfloor\frac{k_{\max}}{2}\rfloor}$
are those in $S_H(v_{\lfloor\frac{k_{\max}}{2}\rfloor}+1)$. Their distances have increased by exactly $2$.
On the other hand, the vertices in $T_{v_{\lfloor\frac{k_{\max}}{2}\rfloor}}(v_0)$ are now closer to $v_{\lfloor\frac{k_{\max}}{2}\rfloor}$. 
If $k_{\max}$ is odd then they are exactly $1$ closer and if $k_{\max}$ is even then they are exactly $2$ closer. Therefore, the Nash 
equilibrium conditions imply
\begin{alignat}{2}
|T_{v_{\lfloor\frac{k_{\max}}{2}\rfloor}}(v_0)| &\le 2\cdot |S_H(v_{\lfloor\frac{k_{\max}}{2}\rfloor+1})|   
		&\qquad \text{if } k_{\max} \text{ odd} \label{eq:bi-3-bodd} \\
|T_{v_{\lfloor\frac{k_{\max}}{2}\rfloor}}(v_0)| &\le |S_H(v_{\lfloor\frac{k_{\max}}{2}\rfloor+1})|   
		 &\qquad \text{if } k_{\max}  \text{ even}  \label{eq:bi-3-beven}
\end{alignat}

\

\noindent (iii) $v_{\lceil\frac{k_{\max}}{2}\rceil}$ sells $e_{\lceil\frac{k_{\max}}{2}\rceil}$ and buys $v_{\lceil\frac{k_{\max}}{2}\rceil}v_0$.

This is an instance of Strategy~I with a slight twist for odd $k$. First note that 
$T_{v_0}(v_{\lceil\frac{k_{\max}}{2}\rceil}) = S_H(v_{\lceil\frac{k_{\max}}{2}\rceil})$. 
For even $k$ we then obtain the following bound from~(\ref{eq:strat1}).
$$0 \leq D(v_0) +n-1 - D(v_{\lceil\frac{k_{\max}}{2}\rceil}) -  \left(\left\lfloor\frac{k_{\max}}{2}\right\rfloor+1\right)\cdot|S_H(v_{\lceil\frac{k_{\max}}{2}\rceil})|$$ 
However, we can improve upon this for odd $k$. Note that $S_H(v_{\lfloor\frac{k_{\max}}{2}\rfloor})$ is $(\lfloor\frac{k_{\max}}{2}\rfloor-1)$ 
closer to $v_{\lfloor\frac{k_{\max}}{2}\rfloor}$ than $v_0$ was before the switch. Thus for odd $k$ we use the bound
$$0 \leq D(v_0) +n-1 - D(v_{\lceil\frac{k_{\max}}{2}\rceil}) -  \left(\left\lfloor\frac{k_{\max}}{2}\right\rfloor+1\right)\cdot|S_{H}(v_{\lceil\frac{k_{\max}}{2}\rceil})| 
- \left\lfloor\frac{k_{\max}}{2}\right\rfloor \cdot|S_H(v_{\lfloor\frac{k_{\max}}{2}\rfloor})| $$

These two separate inequalities for odd and even $k$ can be turned into the following bound for all $k$ using the inequalities 
(\ref{eq:bi-3-a}), (\ref{eq:bi-3-bodd}) and (\ref{eq:bi-3-beven}).
\begin{equation}\label{eq:bi-3-c} 
	0 \le D(v_0) - D(v_{\lceil\frac{k_{\max}}{2}\rceil})  + n-1- \left\lfloor\frac{k_{\max}}{2}\right\rfloor\cdot |T_{v_{\lceil\frac{k_{\max}}{2}\rceil}}(v_0)| 
\end{equation}

\

\noindent (iv) $v_0$ sells $e_0$ and $a$ and buys $v_0v_{\lceil\frac{k_{\max}}{2}\rceil}$. 

This is just a straightforward application of Strategy II. Here~(\ref{eq:strat3}) yields
\begin{alignat}{1}\label{eq:bi-3-e}
0 &\le D(v_{\lceil\frac{k_{\max}}{2}\rceil}) -D(v_0) + n-1 + \left(k_{\max} - \left\lfloor\frac{k_{\max}}{2} \right\rfloor -3 \right)\cdot |T_{v_{\lceil\frac{k_{\max}}{2}\rceil}}(v_0)| -\alpha\nonumber \\
&\leq D(v_{\lceil\frac{k_{\max}}{2}\rceil}) -D(v_0) + n-1 + \left\lfloor\frac{k_{\max}}{2} \right\rfloor\cdot |T_{v_{\lceil\frac{k_{\max}}{2}\rceil}}(v_0)|  -\alpha
\end{alignat}

Finally, summing (\ref{eq:bi-3-c}) and (\ref{eq:bi-3-e})  gives $0 \le 2(n-1) - \alpha$.
But this contradicts the assumption that $2(n-1) < \alpha$.
This completes the proof; there exist two vertices $v_1,v_2\in C$ with $d(v,u)\geq \frac{k_{\max}}{3}$ 
which buy edges $f, g\not\in C$ respectively.
\end{proof}

\section{The Tree Conjecture holds for $\alpha>3(n-1)$}

In this section, we put all the pieces together to obtain the following result.

\begin{theorem}
For $\alpha > 3(n-1)$, any Nash equilibrium $G$ is a tree.
\end{theorem}
\begin{proof}
Let $C$ be a maximum length min-cycle of a biconnected component $H$. By \cref{multipleoutedge}, 
there are two vertices $r_1,r_2\in C$ with $d(r_1,r_2)\geq \frac{k_{\max}}{3}$ 
that buy edges $e_1,e_2\in \{C\}$ $f_1,f_2\in H\setminus \{C\}$, respectively. 
Let $P_1$ be the shorter of the two paths in $C$ between $r_1$ and $r_2$. 
Without loss of generality, let $P_1$ be directed from $r_1$ to $r_2$. Finally, take a vertex $r_3\in H\cap T_{r_1}(r_2)$ that is as deep 
as possible in $T_{r_1}(r_2)$.
It must buy an edge $f_3 \in H\setminus T_{r_1}(r_2)$ because all vertices in $H$ are in directed 
min-cycles by \cref{mincycleedges} and~\cref{directedcycles}.

We now consider six cases for which we propose a set of strategy changes.
None of these strategies decrease the agent's cost only if $\alpha \leq  3(n-1)$.
This fact implies no biconnected component $H$ 
can exist when $\alpha > 3(n-1)$, giving the theorem.

\

\noindent{\bf Case 1: $D(r_1) \geq D(x)$ for some $x \in T_{r_1}(r_2)$}\\ 
In this case, we consider two strategy changes:
\begin{enumerate}
	\item $r_1$ sells $e_1$ and buys $r_1x$
	\item $r_1$ sells $e_1$ and $f_1$ and buys $r_1x$
\end{enumerate}

The first change is an instance of Strategy~I. Thus from~(\ref{eq:strat1}), we have the bound:
\begin{equation}\label{eq:1}
	0 \ \leq\  D(x) - D(r_1)  + n-1 - (|P_1|+d(r_2,x)+1)\cdot |T_{x}(r_1)|
\end{equation}

The second change is an instance of Strategy~II. Thus from~(\ref{eq:strat3}), we have the bound:
\begin{equation}\label{eq:2}
	0 \ \leq\  D(x) - D(r_1)  + n-1 +(k_{\max} -|P_1|-d(r_2,x)-3)\cdot |T_{x}(r_1)| - \alpha
\end{equation}

The linear combination $2\times (\ref{eq:1}) +1\times (\ref{eq:2})$ gives
\begin{align*}
	0 &\leq 3\cdot\big(D(x) - D(r_1)  + n-1 -(|P_1|+d(r_2,x))\cdot |T_{x}(r_1)| \big)  +  (k_{\max}-5)\cdot |T_{x}(r_1)| - \alpha \\
	&\leq 3\cdot\big(n-1 -(|P_1|+d(r_2,x))\cdot |T_{x}(r_1)| \big)  +  k_{\max}\cdot |T_{x}(r_1)| - \alpha \\
	&\leq 3\cdot\left(n-1 - \frac13k_{\max}\cdot |T_{x}(r_1)|\right) +  k_{\max}\cdot |T_{x}(r_1)| - \alpha \\
	&\leq 3(n-1) - \alpha
\end{align*}
Here the second inequality holds as $D(r_1) \ge D(x)$.
It follows that for $\alpha > 3(n-1)$, we may now assume $D(r_1) < D(x)$ for every $x \in T_{r_1}(r_2)$.

\

\noindent{\bf Case 2: $d(r_2,r_3)  \geq \frac13k_{\max}$}\\
In this case, we consider three strategy changes:
\begin{enumerate}
	\item $r_1$ sells $e_1$ and buys $r_1r_3$
	\item $r_1$ sells $e_1$ and $f_1$ and buys $r_1r_3$
	\item $r_3$ sells $f_3$  and buys $r_1r_3$
\end{enumerate}

The first change is an instance Strategy~I. Thus from~(\ref{eq:strat1}), we have the bound:
\begin{equation}\label{eq:3}
	0 \ \leq\  D(r_3) - D(r_1)  + n-1 -(|P_1|+d(r_3,r_2)+1) \cdot |T_{r_3}(r_1)|
\end{equation}

The second change is an instance of Strategy~II. Thus from~(\ref{eq:strat3}), we have the bound:
\begin{equation}\label{eq:4}
	0 \ \leq\  D(r_3) - D(r_1)  + n-1 +(k_{\max} -|P_1|-d(r_3,r_2)-3) \cdot |T_{r_3}(r_1)| - \alpha
\end{equation}

The third change  is an instance Strategy~I. Thus from~(\ref{eq:strat1}), we have the bound:
 \begin{alignat}{1}\label{eq:5}
 	0 & \leq\  D(r_1) - D(r_3) + n-1 -(|P_1|+d(r_3,r_2)+1) \cdot |T_{r_1}(r_3)| \nonumber\\
 	& \leq\ D(r_1) - D(r_3) + n-1 
 \end{alignat}
The linear combination $\frac12 \times (\ref{eq:3})+ 1\times (\ref{eq:4})+ \frac32\times (\ref{eq:5})$, all $D$ terms cancel and we're left with
\begin{align*}
	0 &\leq \frac32\cdot \big(n-1 + n-1 -(|P_1|+d(r_3,r_2)) \cdot |T_{r_3}(r_1)|\big) +  (k_{\max}-3.5)\cdot |T_{r_3}(r_1)| - \alpha \\
	&= 3(n-1)- \frac32\cdot  (|P_1|+d(r_3,r_2)) \cdot |T_{r_3}(r_1)|  +  k_{\max}\cdot |T_{r_3}(r_1)| - \alpha \\
	&\le 3(n-1)- k_{\max} \cdot |T_{r_3}(r_1)|  +  k_{\max}\cdot |T_{r_3}(r_1)| - \alpha \\
	&\leq 3(n-1) - \alpha
\end{align*}
Here the third inequality holds as $d(r_2,r_3)  \geq \frac13k_{\max}$.
It follows that for $\alpha > 3(n-1)$, we may now assume that $d(r_2,r_3)  < \frac13k_{\max}$.

\

\noindent{\bf Case 3: $|T_{r_1}(r_2)| \leq  \frac{n-1}{d(r_3,r_2)}$}\\
In this case, we consider one new strategy change:
\begin{enumerate}
	\item $r_2$ sells $e_2$ and $f_2$ and buys $r_1r_2$
\end{enumerate}

This strategy change is an instance of Strategy~III. Thus from~(\ref{eq:strat4}), we have the bound:
\begin{align*}\label{eq:6}
0 &\leq D(r_1) - D(r_2)  + n-1  +  2d(r_3,r_2)\cdot|T_{r_1}(r_2)| - \alpha\\
&< n-1  +  2d(r_3,r_2)\cdot|T_{r_1}(r_2)| - \alpha\\
&\leq n-1  +  2(n-1)- \alpha\\
	&\leq 3(n-1) - \alpha 
\end{align*}
Here the strict inequality holds as $D(r_1)< D(r_2)$; the second inequality holds by 
the assumption $|T_{r_1}(r_2)| \leq  \frac{n-1}{d(r_3,r_2)}$.
Thus, for $\alpha > 3(n-1)$, we may now assume that $|T_{r_1}(r_2)| >  \frac{n-1}{d(r_3,r_2)}$.

\

\noindent{\bf Case 4: $D(r_1)- D(r_3) \leq -(2-\frac{|P_1|}{d(r_3,r_2)})\cdot (n-1)$}\\
In this case, we reconsider one strategy change:
\begin{enumerate}
	\item $r_3$ sells $f_3$  and buys $r_1r_3$
\end{enumerate}
Recall $d(r_2,r_3)  < \frac13k_{\max}$.
Then because $|P_1| \geq \frac13k_{\max}$, $u$ 
is closer to every vertex of $ |T_{r_1}(r_2)|$ than $r_1$ by at least $|P_1|-d(r_3,r_2)$.
Thus, the equilibrium conditions imply that
\begin{align*}
0 &\leq  D(r_1) - D(r_3)  + n-1  -(|P_1|-d(r_3,r_2)) \cdot |T_{r_1}(r_2)| \\
&\leq -(2-\frac{|P_1|}{d(r_3,r_2)})\cdot (n-1) + n-1  -(|P_1|-d(r_3,r_2)) \cdot |T_{r_1}(r_2)| \\
&< -(2-\frac{|P_1|}{d(r_3,r_2)})\cdot(n-1) + n-1  -(|P_1|-d(r_3,r_2)) \cdot \frac{n-1}{d(r_3,r_2)}\\
	&=0
\end{align*}
Thus, for $\alpha > 3(n-1)$, we may now assume 
$D(r_1)- D(r_3) > -\left(2-\frac{|P_1|}{d(r_3,r_2)}\right)\cdot(n-1)$ or, equivalently, $D(r_3) - D(r_1) < \left(2-\frac{|P_1|}{d(r_3,r_2)}\right)\cdot(n-1)$.

\

\noindent{\bf Case 5: $ |T_{r_3}(r_1)| \geq \frac{\left(3-\frac{|P_1|}{d(r_3,r_2)}\right)\cdot (n-1)}{(|P_1|+d(r_3,r_2))}$}\\
In this case, we reconsider one strategy change:
\begin{enumerate}
	\item $r_1$ sells $e_1$ and buys $r_1r_3$
\end{enumerate}
This is an instance Strategy~I. Thus from~(\ref{eq:strat1}), we have the bound:
\begin{align*}
0 &\leq  D(r_3) - D(r_1) + n-1  - \big(|P_1|+d(r_3,r_2)+1\big) \cdot |T_{r_3}(r_1)| \\
&<  \left(2-\frac{|P_1|}{d(r_3,r_2)}\right)\cdot (n-1)  + n-1  - \big(|P_1|+d(r_3,r_2)\big) \cdot |T_{r_3}(r_1)| \\
&\le \left(3-\frac{|P_1|}{d(r_3,r_2)}\right)\cdot (n-1)   -(|P_1|+d(r_3,r_2)) \cdot   \frac{\left(3-\frac{|P_1|}{d(r_3,r_2)}\right)\cdot (n-1)}{(|P_1|+d(r_3,r_2))}\\
	&=0
\end{align*}
Thus, for $\alpha > 3(n-1)$, we may now assume  $ |T_{r_3}(r_1)| < \frac{\left(3-\frac{|P_1|}{d(r_3,r_2)}\right)\cdot (n-1)}{(|P_1|+d(r_3,r_2))}$.

\

\noindent{\bf Case 6: Else.}\\
For this final case, define $\gamma= \frac{d(r_3,r_2)}{|P_1|}$. Since $d(r_2,r_3)  < \frac13k_{\max} \le |P_1|$, it follows that $t \in (0,1)$. 
We now consider one strategy change:
\begin{enumerate}
	\item $r_1$ sells $e_1$ and $f_1$ and buys $r_1r_3$
\end{enumerate}
This is an instance Strategy~II. Thus from~(\ref{eq:strat3}), we have the bound:
\begin{equation}\label{eq:9}
	0 \ \leq\  D(r_3) - D(r_1)  + n-1 +(k_{\max}-(|P_1|+d(r_2,r_3)+3))\cdot |T_{r_3}(r_1)| - \alpha
\end{equation}
Substituting in $D(r_3) - D(r_1) < (2-\frac{|P_1|}{d(r_3,r_2)})\cdot(n-1)$ and $|T_{r_3}(r_1)| < \frac{(3-\frac{|P_1|}{d(r_3,r_2)})\cdot (n-1)}{(|P_1|+d(r_3,r_2))}$ gives
\begin{align*}
	0 &< \left(3-\frac{|P_1|}{ d(r_3,r_2)}\right)\cdot (n-1) +(k_{\max}-(|P_1|+d(r_2,r_3)+3))\cdot  \frac{(3-\frac{|P_1|}{d(r_3,r_2)})\cdot (n-1)}{(|P_1|+d(r_3,r_2))} - \alpha\\	
	&= (k_{\max}-3)\cdot  \frac{(3-\frac{|P_1|}{d(r_3,r_2)})}{(|P_1|+d(r_3,r_2))}\cdot (n-1) - \alpha\\	
	&\leq  \frac{k_{\max}}{|P_1|}\cdot  \frac{3-\frac{1}{\gamma}}{1+\gamma}\cdot (n-1) - \alpha\\	
	& \leq 3(n-1) \cdot  \frac{3-\frac{1}{\gamma}}{1+\gamma}- \alpha\\	
	& \leq 3(n-1) - \alpha
\end{align*}
The third inequality holds as $\frac13k_{\max} \le |P_1|$.
For the last inequality, we claim that $\frac{3-\frac{1}{\gamma}}{1+\gamma} \leq 1$, for any~$\gamma$.
To see this, note that $(\gamma-1)^2\ge 0$. Rearranging gives $\gamma^2+\gamma\ge 3\gamma-1$ and the claim holds.

This completes the case analysis. If  $\alpha > 3(n-1)$ there are no cases where a biconnected component of $G$ can exist. 
Thus any Nash equilibrium $G$ must be a tree for $\alpha > 3(n-1)$.
\end{proof}

\section{Conclusion}

In this paper, we have shown that the revised tree conjecture holds for $\alpha \ge 3(n-1)$.
Moreover, we have confirmed that min-cycles are a powerful tool in tackling the conjecture. 
Specifically, examining the strategic options of vertices on maximum length 
min-cycles is a promising technique for contradicting the existence of cycles in Nash equilibria. 
This is particularly the case for the range $\alpha > 2(n-1)$, where we know that all min-cycles 
must be directed. Indeed, all the results presented in this paper, except for the main theorem, 
hold for the range $\alpha > 2(n-1)$. This suggests improved bounds can be obtained 
using these methods.

\bibliographystyle{plain}
\bibliography{ncbib}

%
%
%

\end{document}